\documentclass[english]{IEEEtran}

\usepackage{graphicx}
\newtheorem{theorem}{Theorem}
\newtheorem{lemma}{Lemma}

\newtheorem{case}{Case}
\newtheorem{example}{Example}

\newenvironment{proof}[1][Proof]{\begin{trivlist}
\item[\hskip \labelsep {\bfseries #1}]}{\end{trivlist}}

\usepackage{subfigure}
\usepackage{dsfont}

\usepackage{cite}

\ifCLASSINFOpdf
\else
\fi
\usepackage{amsmath}

\usepackage{amsfonts}
\usepackage{acronym}

\usepackage{array}

\usepackage{mdwmath}
\usepackage{mdwtab}

\usepackage{eqparbox}

\begin{document}
%
\title{Optimization of Signal-to-Noise-and-Distortion Ratio for Dynamic Range Limited Nonlinearities}
%
%
%

\author{Kai~Ying,
        Zhenhua~Yu,~\IEEEmembership{Student Member,~IEEE,}
        Robert~J.~Baxley,~\IEEEmembership{Senior Member,~IEEE,}
        and~G.~Tong~Zhou,~\IEEEmembership{Fellow,~IEEE}
\thanks{Kai Ying, Zhenhua Yu and G. Tong Zhou are with the School
of Electrical and Computer Engineering, Georgia Institute of Technology, Atlanta,
GA, 30332 USA (e-mail: kying3@gatech.edu).}
\thanks{Robert J. Baxley is with the Georgia Tech Research Institute.}}%

\maketitle

\begin{abstract}
Many components used in signal processing and communication applications, such as power amplifiers and analog-to-digital converters, are nonlinear and have a finite dynamic range.
The nonlinearity associated with these devices distorts the input, which can degrade the overall system performance.
Signal-to-noise-and-distortion ratio (SNDR) is a common metric to quantify the performance degradation. One way to mitigate nonlinear distortions is by maximizing the SNDR.
In this paper, we analyze how to maximize the SNDR of the nonlinearities in optical wireless communication (OWC) systems.
Specifically, we answer the question of how to optimally predistort a double-sided memory-less nonlinearity that has both a ``turn-on'' value and a maximum ``saturation'' value.
We show that the SNDR-maximizing response given the constraints is a double-sided limiter with a certain linear gain and a certain bias value. Both the gain and the bias are functions of the probability density function (PDF) of the input signal and the noise power.
We also find a lower bound of the nonlinear system capacity, which is given by the SDNR and an upper bound determined by dynamic signal-to-noise ratio (DSNR).
An application of the results herein is to design predistortion linearization of nonlinear devices like light emitting diodes (LEDs).

\end{abstract}


%
\IEEEpeerreviewmaketitle

\section{Introduction}
%
%
%
%
%
%

In addition to being nonlinear, many components in a signal processing or communication system have a dynamic range constraint.
For example, light emitting diodes (LEDs) are dynamic range constrained devices that appear in intensity modulation  (IM) and direct detection (DD) based optical wireless communication (OWC) systems~\cite{IEEE_MAG}\cite{SP01}.
To drive an LED, the input electric signal must be positive and exceed the turn-on voltage of the device.
On the other hand, the signal is also limited by the saturation point or maximum permissible value of the LED.
Thus, the dynamic range constraint can be modeled as two-sided clipping.
The same situation may happen in other applications such as digital audio processing~\cite{Audio}.

Both nonlinearity and clipping result in distortions which may cause system performance degradation.
SNDR is a commonly used metric to quantify the distortion that is uncorrelated to the signal~\cite{RF1}-\cite{RF4}.
Previous work in this area mainly concentrated on a family of amplitude-limited nonlinearities that is common in radio frequency (RF) system design involving nonlinear components such as power amplifiers (PAs) and mixers.

Different from the previous work, our study discusses the class of nonlinearities with a two-sided dynamic range constraint that is more commonly found in optical and acoustic systems.
Authors in~\cite{VLC1}-\cite{VLC4} illustrated the impact of LED nonlinearity and clipping noise in OWC systems.
Some predistortion strategies were proposed in~\cite{VLC5}-\cite{VLC7}.
However, to the best of our knowledge, the optimal nonlinear mapping under the two-sided dynamic range constraint has not been studied.

There are two major differences from the amplitude-limited nonlinearity.
First, the signal will be subject to turn-on clipping and saturation clipping to meet the dynamic range constraint.
Second, DC biasing must be used to shift the signal to an appropriate level to minimize distortion.
In this paper, we will show that the ideal linearizer that maximizes the SNDR is a double-sided limiter that has an affine response. The parameters of the response can be calculated from the distribution of the input signal and the noise power.

In additional to deriving the SNDR-optimal predistorter, we also relate a lower bound on channel capacity to the SNDR, further motivating the SNDR considerations.
Finally, we employ another common distortion metric, dynamic signal-to-noise ratio (DSNR) to provide an upper bound on the double-sided clipping channel.

The remainder of this paper is organized as follows:
Section II introduces the system model for dynamic range limited nonlinearity and the corresponding SNDR definition.
In Section III, we derive the optimal nonlinear mapping that maximizes the SNDR and illustrate some examples.
In section IV, we related the SNDR to the capacity of the nonlinear channel.
Finally, Section VII concludes the paper.
The detailed proofs of this paper are deferred to the Appendices.

\section{System Model and SNDR Definition}

\subsection{System Model}

Let us consider a system modeled by
\begin{equation}
y_o(t)=h_o(x_o(t)) + v(t)
\end{equation}
where $x_o(t)$ is a real-valued signal with mean $\mu_{x}$ and variance $\sigma_x^2$;
$v(t)$ is a zero-mean additive noise process with variance $\sigma_v^2$;
$h_o(\cdot)$ is a memoryless nonlinear mapping with dynamic range constraint $A_1 \leq h_o(x_o(t)) \leq A_2$.

For notational simplicity, we omit the $t$-dependence in the memoryless system and replace $h_o(\cdot)$ and $x_o(t)$ by $h(\cdot)=h_o(\cdot)-A_1$ and $x=x_o-\mu_{x}$.
Then we have an equivalent system modeled by
\begin{equation}\label{mapping}
y=h(x) + v
\end{equation}
where $h(\cdot)$ is a memoryless nonlinear mapping with dynamic range constraint $0 \leq h(x) \leq A=A_2-A_1$ and $x$ is a zero-mean signal with variance $\sigma_x^2$.

\subsection{SNDR Definition}

According to Bussgang's Theorem~\cite{bussgang}, the nonlinear mapping in (\ref{mapping}) can be decomposed as
\begin{equation}
h(x) = \alpha x+d
\end{equation}
where $d$ is the distortion caused by $h(\cdot)$ and $\alpha$ is a constant, selected so that $d$ is uncorrelated with $x$, i.e., $E[xd]=0$.
Thus

\begin{equation}
\alpha = \frac{E[xh(x)]-E[xd]}{E[x^2]}=\frac{E[xh(x)]}{E[x^2]}=\frac{E[xh(x)]}{\sigma_x^2}.
\end{equation}
The distortion power is given by
\begin{equation}
\begin{split}
\varepsilon_d &= E[d^2] - (E[d])^2\\
&= E[h^2(x)]-\alpha^2\sigma_x^2-E^2[h(x)].
\end{split}
\end{equation}
The signal-to-noise-and-distortion ratio (SNDR) is defined as
\begin{equation}\label{SNDRh}
\begin{split}
\mathrm{SNDR} &= \frac{\alpha^2\sigma_x^2}{\varepsilon_d + \sigma_v^2}\\
&= \frac{(E[xh(x)])^2/\sigma_x^2}{E[h^2(x)]-(E[xh(x)])^2/\sigma_x^2-E^2[h(x)] + \sigma_v^2}.
\end{split}
\end{equation}

The definition of SNDR here is a little bit different from that in~\cite{RF4}, because all the signals are real and the distortion contains DC biasing.
Thus, the distortion power is modeled as variance rather than the secondary moment.

We see from (\ref{SNDRh}) that the SNDR is related to the distribution of $x$, the noise power $\sigma_v^2$ and the nonlinear mapping $h(\cdot)$.
Our aim in the next section is to determine the function $h(\cdot)$ that maximizes the SNDR given a signal distribution and the two-sided clipping constraint.

\section{SNDR Optimization and Examples}

\subsection{Optimization of SNDR}
Similar to~\cite{RF4}, let us use a function $g(\cdot)$ to normalize the nonlinear mapping $h(\cdot)$:
\begin{equation}\label{gfunction}
h(x) = Ag\left(\frac{x}{\sigma_x}\right)
\end{equation}
where $0 \leq g(\cdot) \leq 1$. Let $\gamma = x/\sigma_x$ and substitute (\ref{gfunction}) into
(\ref{SNDRh}), we obtain
\begin{equation}\label{SNDRg}
\begin{split}
\mathrm{SNDR}&=\frac{E^2[\gamma g(\gamma)]}{E[g^2(\gamma)]-E^2[\gamma g(\gamma)]-E^2[g(\gamma)] + \sigma_v^2/A^2}\\
&=\frac{E^2[\gamma g(\gamma)]}{var[g(\gamma)]-E^2[\gamma g(\gamma)] + \sigma_v^2/A^2}
\end{split}
\end{equation}
where $var[g(\gamma)]$ is the variance of $g(\gamma)$ and $var[g(\gamma)]=E[g^2(\gamma)]-E^2[g(\gamma)]$.

The SNDR optimization problem can be stated as follows:
\begin{eqnarray}
 &\smash{\displaystyle\max_{g(\cdot)}}& {\mathrm{SNDR}}\\
 &s.t.&  0 \leq g(\cdot) \leq 1
\end{eqnarray}
for a given distribution of $\gamma$, dynamic range $A$ and noise power $\sigma_v^2$.

\begin{figure}[htbp]
\centering
\includegraphics[width=3.5in]{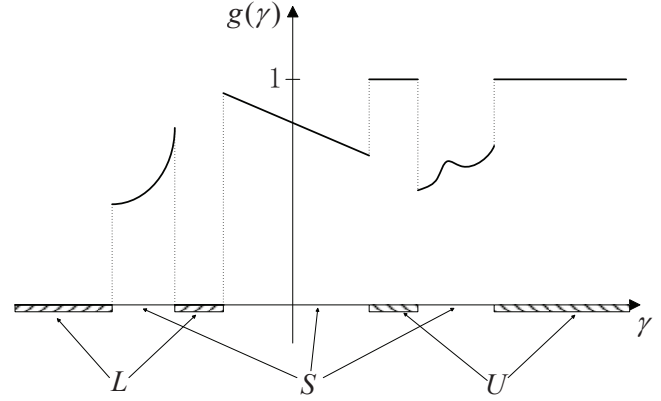}
\caption{An example of nonlinear mapping $g(\cdot)$ that satisfies the $0 \leq g(\cdot) \leq 1$ constraint.}
\label{gexample}
\end{figure}

Fig.~\ref{gexample} illustrates an example of the $g(\cdot)$.
The region of $\gamma$ is divided into three sets $L$, $S$ and $U$.
\begin{eqnarray}
g(\gamma) = 0, \quad \mathrm{for} \quad \gamma \in L;\\
0 < g(\gamma) < 1, \quad \mathrm{for} \quad \gamma \in S; \label{Sdefinition}\\
g(\gamma) = 1, \quad \mathrm{for} \quad \gamma \in U.
\end{eqnarray}
Thus, to determine a nonlinear mapping $g(\cdot)$, we need to find the sets $L$, $S$, $U$ and the shape of the function $g(\cdot)$ in $S$.

We will solve this problem with the following steps:

\begin{enumerate}
    \item find the optimal $g(\cdot)$ given $L$, $S$, $U$;
    \item show that $S$ should be as large as possible;
    \item determine $L$ and $U$ for the optimal solution.
\end{enumerate}

\begin{lemma}
Assume that the sets $L$, $S$ and $U$ are known, and $L\cup S\cup U = R$.
The $g(\cdot)$ function that maximizes the SNDR expression in (\ref{SNDRg}) is of the form
\begin{equation}
g(\gamma) = \frac{\gamma}{\eta} + \beta
\end{equation}
where

\begin{equation}\label{etan}
\begin{split}
\eta &= \frac{C_0^UC_1^S + C_1^U - C_0^SC_1^U}{C_0^U-C_0^UC_0^S-(C_0^U)^2 + (1-C_0^U)\sigma_v^2/A^2}\\
&= \frac{C_0^UC_1^S + C_1^U - C_0^SC_1^U}{C_0^UC_0^L + (1-C_0^S)\sigma_v^2/A^2},
\end{split}
\end{equation}

\begin{equation}\label{beta}
\beta = \frac{C_0^UC_1^S + C_0^UC_1^U + C_1^S\sigma_v^2/A^2}{C_0^UC_1^S + C_1^U -C_0^SC_1^U}
\end{equation}
with

\begin{equation}\label{Cdenotation}
C_{num}^{set} = E[\gamma^{num} I_{set}(\gamma)]
\end{equation}
and $I_{set}(\gamma)$ is the indicator function:
\begin{equation}
\begin{split}
I_{set}(\gamma)=\quad
\begin{cases}
\,\,1, & \text{if } \gamma \in set, \\
\,\,0, & \text{otherwise}.
\end{cases}
\end{split}
\end{equation}
This lemma holds if and only if $S$ satisfies $0 < \frac{\gamma}{\eta} + \beta < 1$ for all $\gamma \in S$.
\end{lemma}

\begin{proof}
See Appendix~\ref{L1proof}.
\end{proof}

This result rules out the $g(\cdot)$ functions whose shape over $S$ is nonlinear.
Fig~\ref{fig2:subfig} demonstrates examples of $g(\cdot)$ functions that may satisfy \emph{Lemma 1}.
Here, the slope of the linear curve in $S$ can be either positive or negative.

\begin{figure}[htbp]
\centering
\subfigure[$\eta > 0$]{
\label{fig2:subfig:a} 
\includegraphics[width=3in]{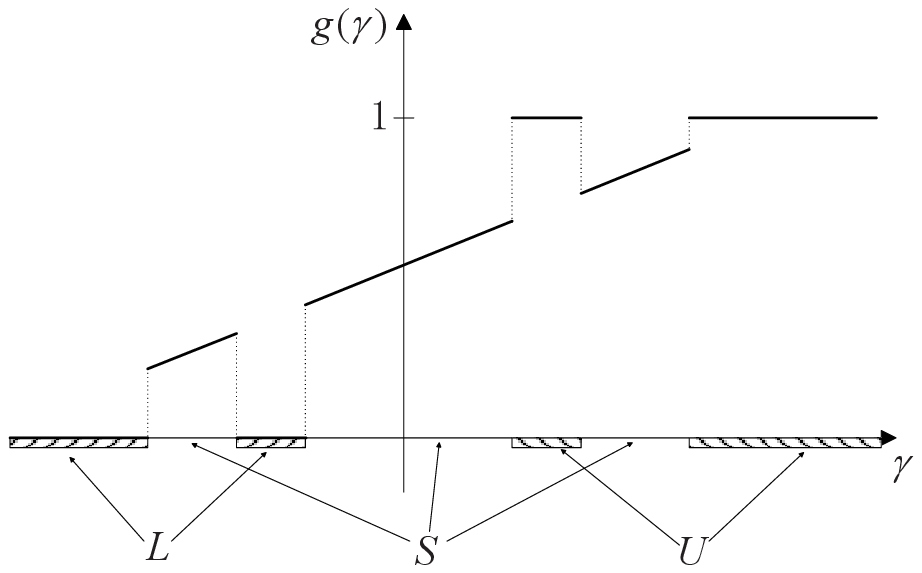}}
\hspace{0.2in}
\subfigure[$\eta < 0$]{
\label{fig2:subfig:b} 
\includegraphics[width=3in]{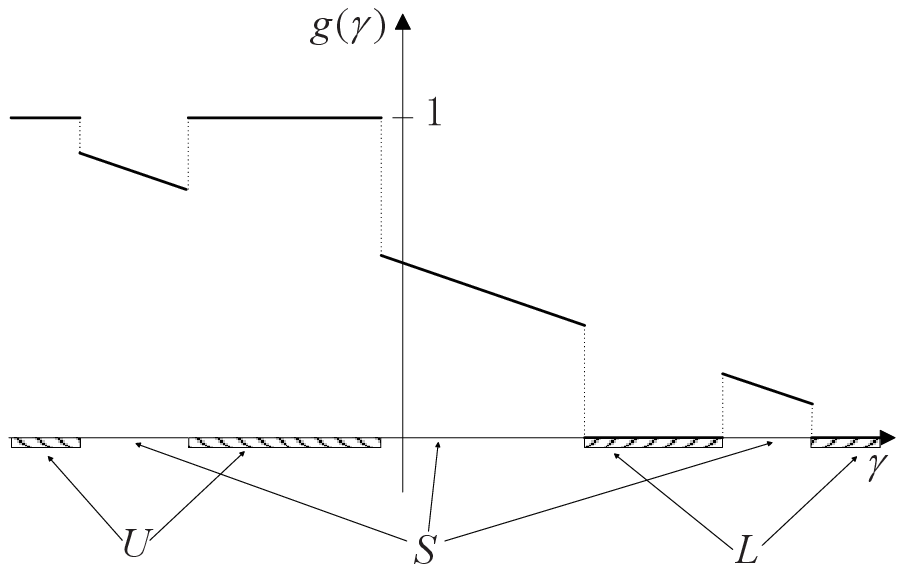}}
\caption{Examples of nonlinear mapping $g(\cdot)$ that may satisfy \emph{Lemma 1}}
\label{fig2:subfig} 
\end{figure}

\emph{Lemma 1} answered the question pertaining to the best shape of the $g(\cdot)$ function with given $L$, $S$ and $U$.
The remaining question is how to determine the optimal sets $L$, $S$ and $U$ so that the SNDR is maximum.
This turns out to be a very challenging problem since we are seeking joint optimization over multiple sets.
Let us consider $S$ first.

\begin{lemma}
Given sets $L$, $S$ and $U$, if $S$ can be enlarged to $S^{\ast}$ such that $S \subset S^{\ast} \subseteq (-\beta^{\ast}\eta^{\ast},\eta^{\ast}-\beta^{\ast}\eta^{\ast})$ or $(\eta^{\ast}-\beta^{\ast}\eta^{\ast},-\beta^{\ast}\eta^{\ast})$, then a higher SNDR can be achieved.
\end{lemma}

\begin{proof}
See Appendix~\ref{L2proof}.
\end{proof}

Fig.~\ref{fig3:subfig} shows how \emph{Lemma 2} works.
$S$ can be enlarged by occupying the subsets of $L$ and $U$.
The larger the set $S$, the better the SNDR that can be achieved.
Just as \emph{Lemma 1}, \emph{Lemma 2} holds if and only if $S^{\ast}$ satisfies
$0 < \frac{\gamma}{\eta^{\ast}} + \beta^{\ast} < 1$ for all $\gamma \in S^{\ast}$,
that is, $S^{\ast} \subseteq (-\beta^{\ast}\eta^{\ast},\eta^{\ast}-\beta^{\ast}\eta^{\ast})$
or $(\eta^{\ast}-\beta^{\ast}\eta^{\ast},-\beta^{\ast}\eta^{\ast})$.

\begin{figure}[htbp]
\centering
\subfigure[$L$, $S$, $U$]{
\label{fig3:subfig:a} 
\includegraphics[width=3in]{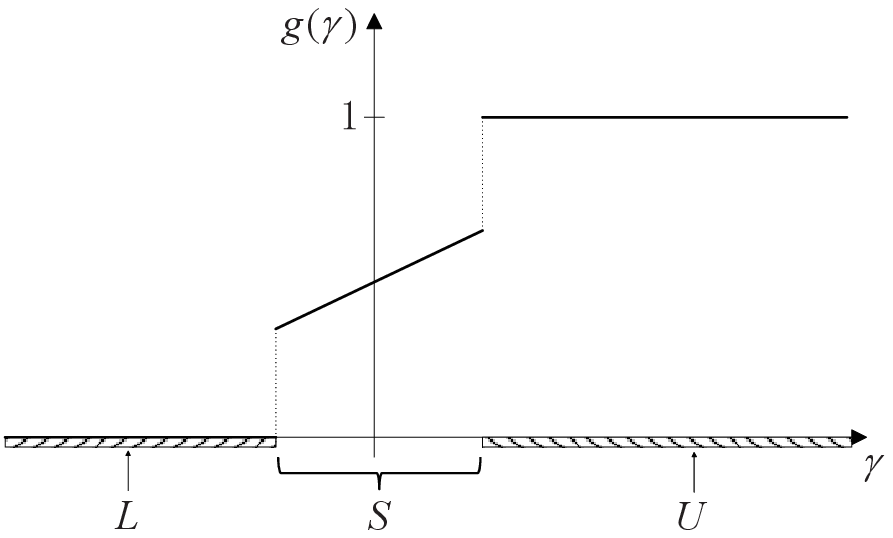}}
\hspace{0.2in}
\subfigure[$L^{\ast}=L-\Delta{L}$, $S^{\ast}=S+\Delta{L}+\Delta{U}$, $U^{\ast}=U-\Delta{U}$]{
\label{fig3:subfig:b} 
\includegraphics[width=3in]{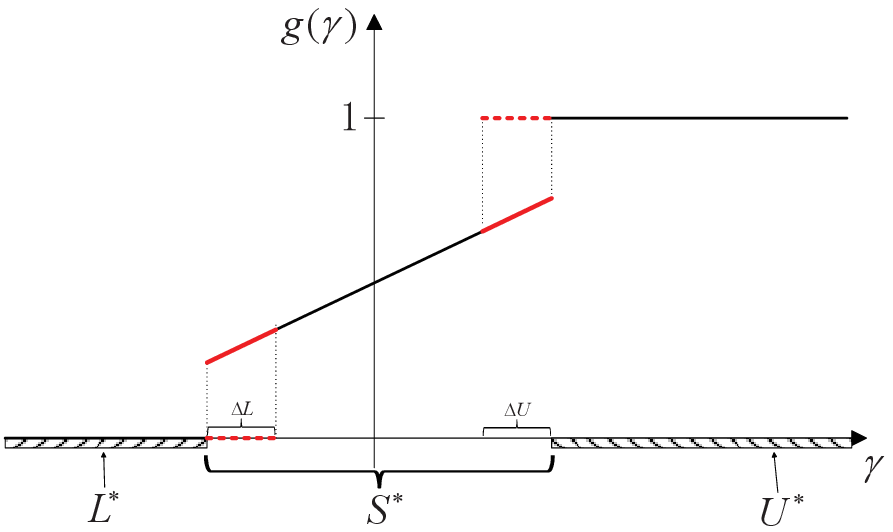}}
\caption{Illustration of \emph{Lemma 2}}
\label{fig3:subfig} 
\end{figure}

Even with the set $S$ determined, we still need to determine $L$ and $U$.

\begin{lemma}
If $\eta>0$, the $g(\cdot)$ that maximizes the SNDR satisfies $L \subset \mathds{R}^-$ and $U \subset \mathds{R}^+$;
if $\eta<0$, the $g(\cdot)$ that maximizes the SNDR satisfies $L \subset \mathds{R}^+$ and $U \subset \mathds{R}^-$.
\end{lemma}

\begin{proof}
Let us compare the SNDR between Fig.~\ref{fig4:subfig:a} and Fig.~\ref{fig4:subfig:b}.
For $\eta>0$, if there is a subset $\Delta{L}$ of $L$ in $\mathds{R}^+$ or a subset $\Delta{U}$ of $U$ in $\mathds{R}^-$, which is illustrated in Fig.~\ref{fig4:subfig:b}, then we see that $E^2[\gamma g(\gamma)]$ is decreased while the variance of $g(\gamma)$ is increased.
Thus, the $\mathrm{SNDR}=\frac{E^2[\gamma g(\gamma)]}{var[g(\gamma)]-E^2[\gamma g(\gamma)] + \sigma_v^2/A^2}$ of Fig.~\ref{fig4:subfig:b} is less than the SDNR of Fig.~\ref{fig4:subfig:a}.
Similarly, we can draw the same conclusion for the case with $\eta<0$.
\end{proof}

\begin{figure}[htbp]
\centering
\subfigure[$L \subset \mathds{R}^-$, $U \subset \mathds{R}^+$]{
\label{fig4:subfig:a} 
\includegraphics[width=3in]{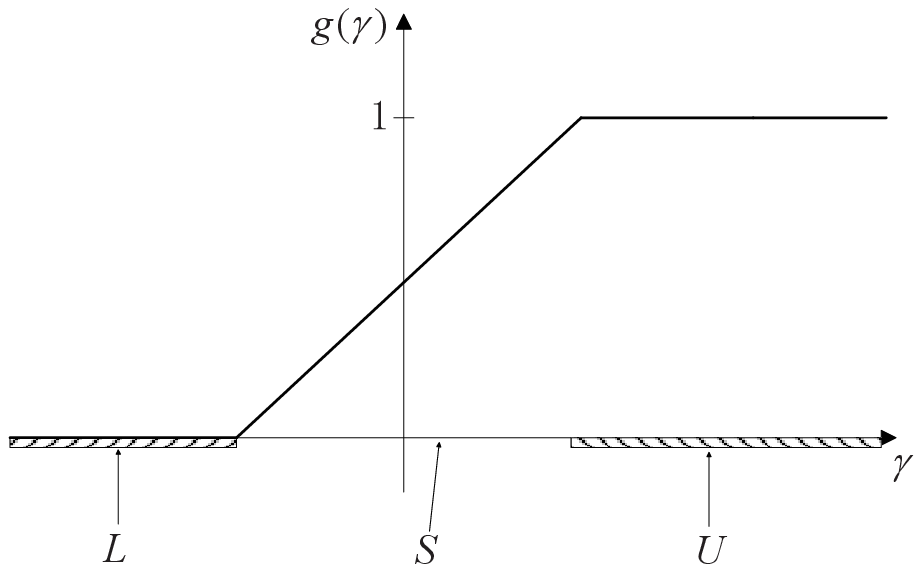}}
\hspace{0.2in}
\subfigure[$\Delta{L} \subset \mathds{R}^+$, $\Delta{U} \subset \mathds{R}^-$]{
\label{fig4:subfig:b} 
\includegraphics[width=3in]{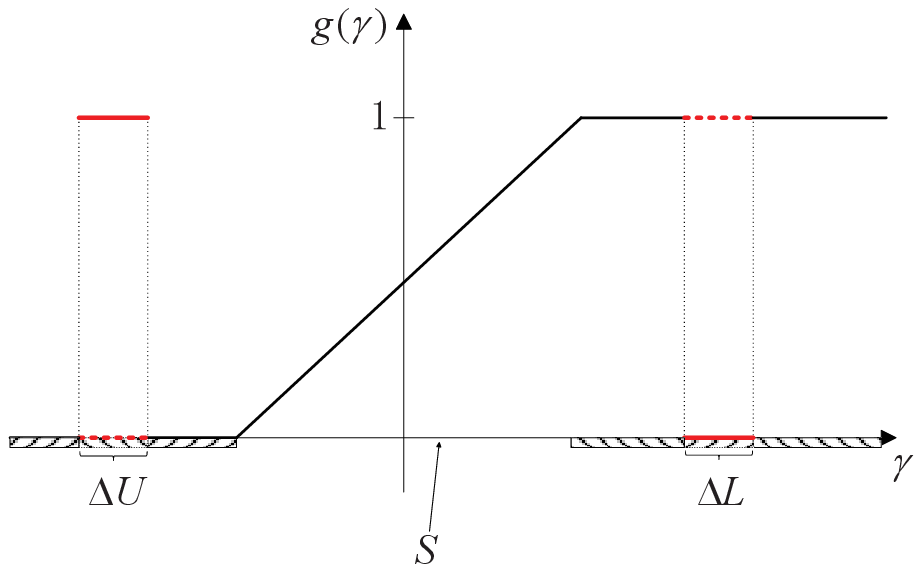}}
\caption{Illustration of \emph{Lemma 3}}
\label{fig4:subfig} 
\end{figure}

In the final analysis, \emph{Lemma 1}, \emph{Lemma 2} and \emph{Lemma 3} imply that the optimal $L$, $S$ and $U$, in the sense of maximizing the SNDR, are $L=(-\infty,-\beta\eta]$, $S=(-\beta\eta,\eta-\beta\eta)$ and $U=[\eta-\beta\eta,+\infty)$ if $\eta>0$;
or $L=[-\beta\eta,+\infty)$, $S=(\eta-\beta\eta,-\beta\eta)$ and $U=[-\infty,\eta-\beta\eta)$ if $\eta<0$.

\begin{theorem}
Within the class of $g(\cdot)$ satisfying $0 \leq g(\cdot) \leq 1$, the following $g(\cdot)$ maximizes the SNDR expression in (\ref{SNDRg}):
\begin{equation}\label{solution1}
\begin{split}
g(\gamma)=\quad
\begin{cases}
\,\,0, &\gamma \leq -\beta^{\star}\eta^{\star}, \\
\,\,\frac{\gamma}{\eta^{\star}}+\beta^{\star}, &-\beta^{\star}\eta^{\star} \leq \gamma \leq \eta^{\star}-\beta^{\star}\eta^{\star},\\
\,\,1, & \gamma \geq \eta^{\star}-\beta^{\star}\eta^{\star}
\end{cases}
\end{split}
\end{equation}
for $\eta^{\star} > 0$, or
\begin{equation}\label{solution2}
\begin{split}
g(\gamma)=\quad
\begin{cases}
\,\,1, &\gamma \leq \eta^{\star}-\beta^{\star}\eta^{\star}, \\
\,\,\frac{\gamma}{\eta^{\star}}+\beta^{\star}, &\eta^{\star}-\beta^{\star}\eta^{\star} \leq \gamma \leq -\beta^{\star}\eta^{\star},\\
\,\,0, & \gamma \geq -\beta^{\star}\eta^{\star}
\end{cases}
\end{split}
\end{equation}
for $\eta^{\star} < 0$, where the $\eta^{\star}$ and $\beta^{\star}$ are found by solving the following transcendental equations:
\begin{equation}\label{eta_star}
\eta^{\star} = \frac{C_0^{U^{\star}}C_1^{S^{\star}} + C_1^{U^{\star}} - C_0^{S^{\star}}C_1^{U^{\star}}}{C_0^{U^{\star}}C_0^{L^{\star}} + (1-C_0^{S^{\star}})\sigma_v^2/A^2},
\end{equation}
\begin{equation}\label{beta_star}
\beta^{\star} = \frac{C_0^{U^{\star}}C_1^{S^{\star}} + C_0^{U^{\star}}C_1^{U^{\star}} + C_1^{S^{\star}}\sigma_v^2/A^2}{C_0^{U^{\star}}C_1^{S^{\star}} + C_1^{U^{\star}} - C_0^{S^{\star}}C_1^{U^{\star}}}
\end{equation}
with

\begin{equation}
\begin{split}
C_0^{U^{\star}}=\quad
\begin{cases}
\,\,\int_{\eta^{\star}-\beta^{\star}\eta^{\star}}^{+\infty}{p(\gamma)d\gamma}, & \text{for } \eta^{\star} > 0, \\
\,\,\int_{-\infty}^{\eta^{\star}-\beta^{\star}\eta^{\star}}{p(\gamma)d\gamma}, & \text{for } \eta^{\star} < 0;
\end{cases}
\end{split}
\end{equation}

\begin{equation}
\begin{split}
C_0^{S^{\star}}=\quad
\begin{cases}
\,\,\int_{-\beta^{\star}\eta^{\star}}^{\eta^{\star}-\beta^{\star}\eta^{\star}}{p(\gamma)d\gamma}, & \text{for } \eta^{\star} > 0, \\
\,\,\int_{\eta^{\star}-\beta^{\star}\eta^{\star}}^{-\beta^{\star}\eta^{\star}}{p(\gamma)d\gamma}, & \text{for } \eta^{\star} < 0;
\end{cases}
\end{split}
\end{equation}

\begin{equation}
\begin{split}
C_0^{L^{\star}}=\quad
\begin{cases}
\,\,\int_{-\infty}^{-\beta^{\star}\eta^{\star}}{p(\gamma)d\gamma}, & \text{for } \eta^{\star} > 0, \\
\,\,\int_{-\beta^{\star}\eta^{\star}}^{\infty}{p(\gamma)d\gamma}, & \text{for } \eta^{\star} < 0;
\end{cases}
\end{split}
\end{equation}

\begin{equation}
\begin{split}
C_1^{U^{\star}}=\quad
\begin{cases}
\,\,\int_{\eta^{\star}-\beta^{\star}\eta^{\star}}^{+\infty}{\gamma p(\gamma)d\gamma}, & \text{for } \eta^{\star} > 0, \\
\,\,\int_{-\infty}^{\eta^{\star}-\beta^{\star}\eta^{\star}}{\gamma p(\gamma)d\gamma}, & \text{for } \eta^{\star} < 0;
\end{cases}
\end{split}
\end{equation}

\begin{equation}
\begin{split}
C_1^{S^{\star}}=\quad
\begin{cases}
\,\,\int_{-\beta^{\star}\eta^{\star}}^{\eta^{\star}-\beta^{\star}\eta^{\star}}{\gamma p(\gamma)d\gamma}, & \text{for } \eta^{\star} > 0, \\
\,\,\int_{\eta^{\star}-\beta^{\star}\eta^{\star}}^{-\beta^{\star}\eta^{\star}}{\gamma p(\gamma)d\gamma}, & \text{for } \eta^{\star} < 0
\end{cases}
\end{split}
\end{equation}
and $p(\gamma)$ is the probability density function (PDF) of $\gamma$. The optimal SNDR is found as
\begin{equation}
\mathrm{SNDR}^{\star}=\frac{1}{\frac{1}{R(\eta^{\star}, \beta^{\star})}-1}
\end{equation}
where
\begin{equation}
R(\eta^{\star}, \beta^{\star})= C_2^{S^{\star}} + \eta^{\star} C_1^{U^{\star}} + \eta^{\star} \beta^{\star} C_1^{S^{\star}}
\end{equation}
and
\begin{equation}
\begin{split}
C_2^{S^{\star}}=\quad
\begin{cases}
\,\,\int_{-\beta^{\star}\eta^{\star}}^{\eta^{\star}-\beta^{\star}\eta^{\star}}{\gamma^2 p(\gamma)d\gamma}, & \text{for } \eta^{\star} > 0, \\
\,\,\int_{\eta^{\star}-\beta^{\star}\eta^{\star}}^{-\beta^{\star}\eta^{\star}}{\gamma^2 p(\gamma)d\gamma}, & \text{for } \eta^{\star} < 0.
\end{cases}
\end{split}
\end{equation}

\end{theorem}

\begin{proof}
See the proofs of \emph{Lemma 1}, \emph{Lemma 2} and \emph{Lemma 3}.
\end{proof}

\emph{Theorem 1} establishes that the nonlinearity in the shape of Fig.~\ref{fig5:subfig} is optimal.

\begin{figure}[htbp]
\centering
\subfigure[$\eta^{\star}>0$]{
\label{fig5:subfig:a} 
\includegraphics[width=2.961in]{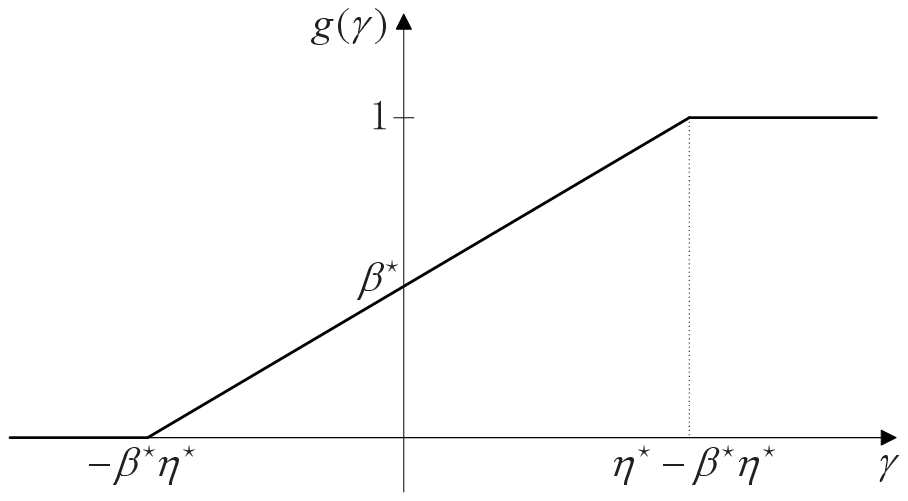}}
\hspace{0.2in}
\subfigure[$\eta^{\star}<0$]{
\label{fig5:subfig:b} 
\includegraphics[width=3.039in]{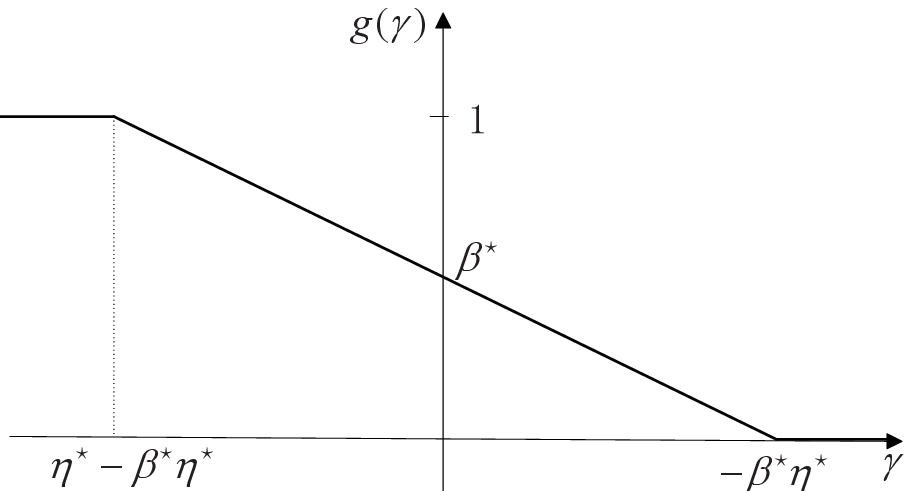}}
\caption{Illustration of optimal $g(\cdot)$ functions to maximize the SNDR}
\label{fig5:subfig} 
\end{figure}

Predistortion is a well-known linearization strategy in many applications such as RF amplifier linearization.
For the dynamic range constrained nonlinearities like LED electrical-to-optical conversion, predistortion has been proposed to mitigate the nonlinear effects.
Specifically, given a system nonlinearity $u(\cdot)$, it is possible to apply a predistortion mapping $f(\cdot)$ so the overall response is linear.
According to \emph{Theorem 1}, it is best to make $u(f(\cdot))$ equal to the $g(\cdot)$ function given in (\ref{solution1}) or (\ref{solution2}) if $u(\cdot)$ is normalized with dynamic range constraint $0 \leq u(\cdot) \leq 1$.
Using the analytical tools presented above, we can answer the questions regarding the selection of the gain factor $1/\eta$, DC biasing $\beta$ and the clipping regions on both sides, or equivalently, the sets $L$ and $U$.
\emph{Theorem 1} shows that these optimal parameters (in terms of SNDR) depend on the PDF of $\gamma$ and the dynamic signal-to-noise ratio $\mathrm{DSNR}=A^2/\sigma_v^2$.
Thus, our work can serve as a guideline for the system design.
In the next subsection, examples are given to illustrate the calculations of the optimal factors $\eta^{\star}$ and $\beta^{\star}$.

\subsection{Examples for selections of optimal parameters}
In the last subsection, we learned that the optimal factors $\eta^{\star}$ and $\beta^{\star}$ can be calculated by solving two transcendental equations (\ref{eta_star}) and (\ref{beta_star}).
However, there may not be closed-form expressions for the solutions.
Additionally, solving (\ref{eta_star}) and (\ref{beta_star}) may result in multiple solutions, but we only keep the real-valued ones since all the signals here are real-valued.

Here, let us take into account a specific class of input signals whose distributions exhibit axial symmetry, such as uniform distribution and Gaussian distribution.
When the distribution of the input signal is axial symmetric, the optimal clipping regions $L^{\star}$ and $U^{\star}$ are also symmetric.
Thus, $C_0^{U^{\star}}=C_0^{L^{\star}}$, $C_1^{U^{\star}}=-C_1^{L^{\star}}$ and $C_1^{S^{\star}}=0$.
Then the factors $\beta^{\star}$ and $\eta^{\star}$ can be calculated:

\begin{equation}
\beta^{\star}=\frac{C_0^{U^{\star}}C_1^{U^{\star}}}{C_0^{U^{\star}}C_1^{U^{\star}}+C_0^{L^{\star}}C_1^{U^{\star}}}=0.5,
\end{equation}

\begin{equation}
\eta^{\star}=\frac{2C_0^{U^{\star}}C_1^{U^{\star}}}{(C_0^{U^{\star}})^2+2C_0^{U^{\star}}\sigma_v^2/A^2}
=\frac{2C_1^{U^{\star}}}{C_0^{U^{\star}}+2\sigma_v^2/A^2}.
\end{equation}

We see that the DC biasing will be the midpoint of the dynamic range.
When the gain factor $\eta^{\star}>0$, it can be further expressed as:

\begin{equation}\label{etastarequ}
\eta^{\star}=\frac{2\int_{0.5\eta^{\star}}^{+\infty}{\gamma p(\gamma)d\gamma}}{\int_{0.5\eta^{\star}}^{+\infty}{p(\gamma)d\gamma}+2\sigma_v^2/A^2}.
\end{equation}

When the gain factor $\eta^{\star}<0$, it can expressed as:

\begin{equation}\label{etastarequneg}
\eta^{\star}=\frac{2\int_{-\infty}^{0.5\eta^{\star}}{\gamma p(\gamma)d\gamma}}{\int_{-\infty}^{0.5\eta^{\star}}{p(\gamma)d\gamma}+2\sigma_v^2/A^2}.
\end{equation}

There is still no closed-form expression for gain factor $\eta^{\star}$.
Next, as examples, let us consider the calculations for uniform distribution and Gaussian distribution specifically.

\begin{example}
When the original signal $x_o(t)$ is uniformly distributed in the interval $[\mu_x-b,\mu_x+b]$,
we infer that the normalized signal $\gamma$ is uniformly distributed in the interval $[-\sqrt{3},\sqrt{3}]$ with the PDF

\begin{equation}
\begin{split}
p(\gamma)=\quad
\begin{cases}
\quad \frac{1}{2\sqrt{3}}, &-\sqrt{3} \leq \gamma \leq \sqrt{3}, \\
\quad 0, &\text{otherwise}.
\end{cases}
\end{split}
\end{equation}

\end{example}

For the case with $\eta^{\star} > 0$, it is straightforward to calculate
\begin{eqnarray}
C_1^{U^{\star}}&=&\int_{0.5\eta^{\star}}^{\sqrt{3}}{\gamma \frac{1}{2\sqrt{3}}d\gamma}=\frac{1}{4\sqrt{3}}(3-\frac{1}{4}{\eta^{\star}}^2)\label{c1uu},\\
C_0^{U^{\star}}&=&\int_{0.5\eta^{\star}}^{\sqrt{3}}{\frac{1}{2\sqrt{3}}d\gamma}=\frac{\sqrt{3}-0.5\eta^{\star}}{2\sqrt{3}}\label{c0uu}.
\end{eqnarray}

Substituting (\ref{c1uu}) and (\ref{c0uu}) into (\ref{etastarequ}), we obtain

\begin{equation}\label{quadratic}
\eta^{\star}=\frac{\frac{1}{2\sqrt{3}}(3-\frac{1}{4}{\eta^{\star}}^2)}{\frac{\sqrt{3}-0.5\eta^{\star}}{2\sqrt{3}}+2\sigma_v^2/A^2}.
\end{equation}

Equation (\ref{quadratic}) can be rewritten as a quadratic equation
\begin{equation}\label{etastarequex1}
{\eta^{\star}}^2-(16\sqrt{3}\sigma_v^2/A^2+4\sqrt{3})\eta^{\star}+12=0.
\end{equation}

Thus, we can obtain a closed-form solution for the optimal $\eta^{\star}$:
\begin{equation}\label{etasolex1}
\eta^{\star}=8\sqrt{3}\sigma_v^2/A^2+2\sqrt{3}-4\sqrt{12\sigma_v^4/A^4+6\sigma_v^2/A^2}.
\end{equation}

We know that there should be two solutions for equation (\ref{etastarequex1}).
In fact, the other solution is $0.5\eta^{\star}>\sqrt{3}$, which means that both $C_0^{U^{\star}}$ and $C_1^{U^{\star}}$ are 0.
Thus, the solution given by (\ref{etasolex1}) is the unique optimal selection for the gain factor $\eta^{\star} > 0$.
If $\eta^{\star}<0$ is desired, the optimal solution is
\begin{equation}\label{etasolex12}
\eta^{\star}=-8\sqrt{3}\sigma_v^2/A^2-2\sqrt{3}+4\sqrt{12\sigma_v^4/A^4+6\sigma_v^2/A^2}.
\end{equation}

\begin{example}
When the original signal $x_o(t)$ is Gaussian distributed,
then the normalized signal $\gamma$ has a standard Gaussian distribution with the PDF

\begin{equation}
p(\gamma)=\frac{1}{\sqrt{2\pi}}e^{-\frac{1}{2}\gamma^2}.
\end{equation}

\end{example}

For the case with $\eta^{\star} > 0$, we have
\begin{eqnarray}
C_1^{U^{\star}}&=&\int_{0.5\eta^{\star}}^{+\infty}{\gamma \frac{1}{\sqrt{2\pi}}e^{-\frac{1}{2}\gamma^2}}=\frac{1}{\sqrt{2\pi}}e^{-\frac{1}{8}{\eta^{\star}}^2}\label{c1ug},\\
C_0^{U^{\star}}&=&\frac{1}{2}-\frac{1}{2} erf(\frac{\eta^{\star}}{2\sqrt{2}})\label{c0ug}
\end{eqnarray}
where $erf(\cdot)$ is the error function with the definition
\begin{equation}
erf(z)=\frac{1}{\sqrt{\pi}}\int_{-z}^z{e^{-\gamma^2}}d\gamma.
\end{equation}

Substituting (\ref{c1ug}) and (\ref{c0ug}) into (\ref{etastarequ}) and simplifying, we obtain

\begin{equation}
\eta^{\star}(\frac{1}{2}-\frac{1}{2} erf(\frac{\eta^{\star}}{2\sqrt{2}})+2\sigma_v^2/A^2)=\frac{2}{\sqrt{2\pi}}e^{-\frac{1}{8}{\eta^{\star}}^2}.
\end{equation}

Here the optimal $\eta^{\star}$ does not have a closed-form expression but can be easily calculated numerically.
We can draw the similar conclusion for the case with $\eta^{\star} < 0$.

\subsection{Numerical results}
Fig.~\ref{eta} shows the optimal $\eta^{\star}$ as a function of DSNR for the above examples.

\begin{figure}[htbp]
\centering
\includegraphics[width=3.5in]{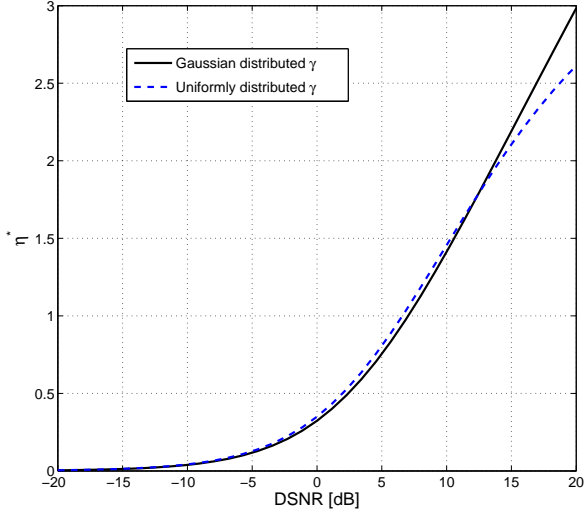}
\caption{Optimal gain factor $\eta^{\star}$ as a function of DSNR for \emph{Example 1} and \emph{Example 2} with $\eta^{\star} > 0$.}
\label{eta}
\end{figure}

Next, we illustrate the SNDR of two different nonlinear mappings.
$g_1(\gamma)$ is the optimal solution chosen by \emph{Theorem 1}. $g_2(\gamma)$ is a fixed mapping given below:

\begin{equation}
\begin{split}
g_2(\gamma)=\quad
\begin{cases}
\,\,0, &\gamma \leq -0.4, \\
\,\,\gamma+0.4, &-0.4 \leq \gamma \leq 0.6,\\
\,\,1, & \gamma \geq 0.6.
\end{cases}
\end{split}
\end{equation}
The corresponding SNDR curves are shown in Fig.~\ref{SNDR}.
This example illustrates that the nonlinearity $g_1(\gamma)$ yields a higher SNDR as compared to the other nonlinearity, as expected according to \emph{Theorem 1}.

\begin{figure}[htbp]
\centering
\includegraphics[width=3.5in]{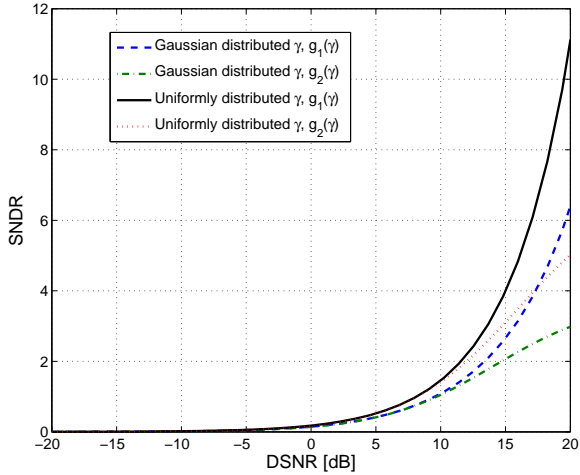}
\caption{SNDR for uniformly and gaussian distributed $\gamma$ with different nonlinear mappings.}
\label{SNDR}
\end{figure}

\section{Relationship Between SNDR and Capacity}

\subsection{Lower Bound on Capacity}
The capacity is given by
\begin{equation}
C=\max_{p_{x_o}}{I(y_o;x_o)}=\max_{p_{x}}{I(y;x)}
\end{equation}
where $I(y;x) = H(x)-H(x|y)=H(y)-H(y|x)$ is the mutual information between $y$ and $x$~\cite{IT}.
To obtain the capacity of the dynamic range constrained channel, we need to solve the following optimization problem:

\begin{eqnarray}
 &\smash{\displaystyle\max_{p_{x}, h(\cdot)}}& {\quad I(y;x)}\notag \\
 &s.t.& \quad 0 \leq h(\cdot) \leq A
\end{eqnarray}
for a specific zero-mean noise with variance $\sigma_v^2$.
Moreover, it can be simplified as:

\begin{eqnarray}
 &\smash{\displaystyle\max_{p_{x_s}}}& {\quad I(x_s+v;x_s)}\notag \\
 &s.t.& \quad 0 \leq x_s \leq A
\end{eqnarray}
which means that we need to find an input distribution in the interval $[0,A]$ to maximize the mutual information.
Specially, when the noise $v$ is Gaussian, the issue is similar to Smith's work in \cite{smith}.
In this case, if DSNR is low, the capacity is achieved by an equal pair of mass points at $0$ and $A$;
if DSNR is high, the asymptotic capacity is the same as the information rate due to a uniformly distributed input in $[0,A]$~\cite{smith}.

However, in most cases, we are most interested in the achievable data rate given a nonlinear channel mapping with any input and any noise.
Similar to the work in~\cite{RF4}, we obtain a lower bound on the information rate:

\begin{eqnarray}
& & I(y;x) \notag \\
&\geq& H(x)-\frac{1}{2}\log(2\pi e \sigma_x^2)+\frac{1}{2}\log\left(\frac{\sigma_y^2}{\sigma_y^2-\frac{\sigma_{xy}^2}{\sigma_x^2}}\right) \\
&=& H(x)-\frac{1}{2}\log(2\pi e \sigma_x^2) \notag\\
 &+&\frac{1}{2}\log\left(\frac{\frac{A^2}{\sigma_v^2} var[g(\gamma)]+1}{\frac{A^2}{\sigma_v^2} var[g(\gamma)]+1-\frac{A^2}{\sigma_v^2}E^2[\gamma g(\gamma)]}\right) \notag \\
&=& H(x)-\frac{1}{2}\log(2\pi e \sigma_x^2)+\frac{1}{2}\log(1+\mathrm{SNDR})
\end{eqnarray}
by referring to (\ref{SNDRg}).
Since $C \geq I(y;x)$ for any input distribution $p_x$, by setting $p_x$ to be the PDF of a zero-mean Gaussian r.v., we obtain

\begin{eqnarray}
C \geq \frac{1}{2}\log(1+\mathrm{SNDR})
\end{eqnarray}
with the SNDR evalutated for a Gaussian $x$.

\subsection{Upper Bound on Capacity}
In this subsection, we find an upper bound for the capacity.
Similar to~\cite{RF4}, supposing $p_y^{\ast}$ is the PDF of $y$ that maximizes the capacity, i.e.,
\begin{eqnarray}
p_y^{\ast}=\arg \max_{p_y}[H(y)-H(y|x)].
\end{eqnarray}

We can write the capacity as
\begin{eqnarray}
C = I(y;x)|_{p_y^{\ast}} &=& H(y)|_{p_y^{\ast}}-H(y|x)\notag\\
&=& H(y)|_{p_y^{\ast}}-H(v)
\end{eqnarray}

Next, we bound the entropy $H(y)$ with the entropy of a Gaussian $y$, yielding
\begin{eqnarray}
C &\leq& \frac{1}{2}\log(2\pi e \sigma_y^2)-H(v)\notag\\
&=& \frac{1}{2}\log(2\pi e \sigma_y^2)-\frac{1}{2}\log(2\pi e \sigma_v^2)+\frac{1}{2}\log(2\pi e \sigma_v^2)-H(v)\notag\\
&=& \frac{1}{2} \log \left(1+\frac{A^2 var[g(\gamma)]}{\sigma_v^2}\right)+\frac{1}{2}\log(2\pi e \sigma_v^2)-H(v) \notag\\
&\leq& \frac{1}{2} \log \left(1+\frac{A^2}{4\sigma_v^2}\right)+\frac{1}{2}\log(2\pi e \sigma_v^2)-H(v)
\end{eqnarray}
where $var[g(\gamma)] \leq \frac{1}{4}$ with $g(\gamma) \in [0,1]$.
Specifically, if the noise is Gaussian, we have the upper bound:
\begin{eqnarray}
C \leq \frac{1}{2}\log\left(1+\frac{A^2}{4\sigma_v^2}\right)
\end{eqnarray}

Since $\varepsilon_d \geq 0$ and $\alpha^2\sigma_x^2 \leq var[h(\gamma)] \leq \frac{1}{4}A^2$, we must have
\begin{equation}
\mathrm{SNDR} = \frac{\alpha^2\sigma_x^2}{\varepsilon_d + \sigma_v^2} \leq \frac{A^2}{4\sigma_v^2}.
\end{equation}
$\frac{A^2}{\sigma_v^2}$ is the defined DSNR which is the same as that in \cite{YU}.

\subsection{Example of Bounds}
Since SNDR is determined by DSNR and the distribution of signal, we plot the bounds as functions of DSNR for Gaussian distributed signal, which is shown in Fig.~\ref{bound}.
We also compare the lower bounds given by two different nonlinear mappings $g_1(\gamma)$ and $g_2(\gamma)$, which are introduced in the last section.
This example illustrates that the nonlinearity $g_1(\gamma)$ chosen according to \emph{Theorem 1} yields a tighter lower bound as compared to the other nonlinearity.
In addition, we can see that the capacity of Gaussian channel as determined by Smith \cite{smith} is between the lower bounds and upper bound that we have.

\begin{figure}[htbp]
\centering
\includegraphics[width=3.5in]{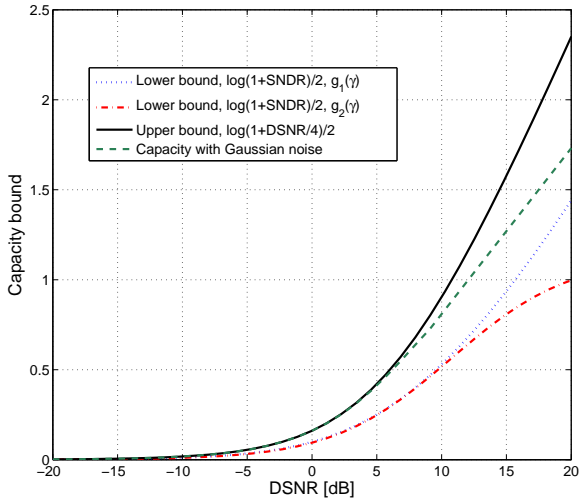}
\caption{Bounds on capacity.}
\label{bound}
\end{figure}

\section{Conclusion}
The main contribution of this paper is the SNDR optimization within the family of dynamic range constrained memoryless nonlinearities.
We showed that, under the dynamic range constraint, the optimal nonlinear mapping that maximizes the SNDR is a double-sided limiter with a particular gain and a particular bias level, which are determined based on the distribution of the input signal and the DSNR.
In addition, we found that $\frac{1}{2}\log(1+\mathrm{SNDR})$ provides a lower bound on the nonlinear channel capacity, and $\frac{1}{2}\log(1+\frac{1}{4}\mathrm{DSNR})$ serves as the upper bound.
The results of this paper can be applied for optimal linearization of nonlinear components and efficient transmission of signals with double-sided clipping.


%

\appendices
\section{Proof of \emph{Lemma 1}}\label{L1proof}
Since we are solving the optimization problem w.r.t. a function, the functional derivative is introduced here~\cite{RF4}\cite{FD}.
By using the Dirac delta function $\delta(\cdot)$ as a test function, the notion of functional derivative is defined as:

\begin{equation}\label{functionalderivative}
\frac{\delta{F[g(\gamma)]}}{\delta{g(\gamma_0)}} = \lim_{\epsilon \rightarrow 0} {\frac{F[g(\gamma)+\epsilon \delta(\gamma-\gamma_0)]-F[g(\gamma)]}{\epsilon}}.
\end{equation}

Just as the variable derivative operation, the linear property, product rule and chain rule hold for functional derivative.
In addition, from (\ref{functionalderivative}), we infer that

\begin{eqnarray}
\frac{\delta{g(\gamma)}}{\delta{g(\gamma_0)}} &=& \delta(\gamma-\gamma_0),\\
\frac{\delta{g^2(\gamma)}}{\delta{g(\gamma_0)}} &=& 2g(\gamma) \delta(\gamma-\gamma_0).
\end{eqnarray}

To maximize the SNDR w.r.t $g(\cdot)$, we need
\begin{equation}
\label{eq:condition}
\frac{\delta \mathrm{SNDR}}{\delta g(\gamma_0)}=0,\quad \forall \gamma_0 \in S.
\end{equation}

We infer that
\begin{equation}\label{Eg}
\begin{split}
E [g(\gamma)] &= E[I_L(\gamma)g(\gamma)] + E[I_S(\gamma)g(\gamma)] + E[I_U(\gamma)g(\gamma)]\\
&= E[I_S(\gamma )g(\gamma)] + E[I_U(\gamma)]\\
&= E[I_S(\gamma )g(\gamma)] + C_0^U.
\end{split}
\end{equation}
Similarly,
\begin{eqnarray}
E [\gamma g(\gamma)] &=& E[I_S(\gamma)\gamma g(\gamma)] + C_1^U \label{EgammaG}, \\
E [g^2(\gamma)] &=& E[I_S(\gamma)g^2(\gamma)] + C_0^U \label{Egs}.
\end{eqnarray}

$C_0^U$ and $C_1^U$ are defined as in (\ref{Cdenotation}). It follows easily that
\begin{eqnarray}
C_0^L+C_0^S+C_0^U &=& 1,\\
C_1^L+C_1^S+C_1^U &=& 0
\end{eqnarray}
and
\begin{equation}
C_0^L, C_0^S, C_0^U  \geq 0.
\end{equation}

Substituting (\ref{Eg}), (\ref{EgammaG}) and (\ref{Egs}) into (\ref{SNDRg})
\begin{equation}
\mathrm{SNDR} = \frac{N[g(\gamma)]}{D[g(\gamma)]}
\end{equation}
where
\begin{equation}
Q[g(\gamma)] = E[I_S(\gamma )\gamma g(\gamma)] + C_1^U,
\end{equation}
\begin{equation}
N[g(\gamma)] = Q^2[g(\gamma)],
\end{equation}
\begin{equation}
Y[g(\gamma)] = E[I_S(\gamma)g(\gamma)] + C_0^U,
\end{equation}
\begin{equation}
\begin{split}
D[g(\gamma)] &= E[I_S(\gamma )g^2(\gamma)] + C_0^U + \frac{\sigma_v^2}{A^2}  \\
&- Q^2[g(\gamma)]- Y^2[g(\gamma)].
\end{split}
\end{equation}
Denote by $p(\gamma)$ the PDF of the random variable $\gamma$. Then
\begin{equation}
E[I_S(\gamma )g^2(\gamma)] = \int I_S(\gamma )g^2(\gamma)p(\gamma)d\gamma.
\end{equation}
Taking the functional derivative w.r.t $g(\gamma_0)$, we obtain
\begin{eqnarray}
& & \frac{\delta E[I_S(\gamma )g^2(\gamma)]}{\delta g(\gamma_0)} \notag \\
&=& \int I_S(\gamma )2g(\gamma)\delta(\gamma-\gamma_0)p(\gamma)d\gamma\\
&=& 2g(\gamma_0)p(\gamma_0).
\end{eqnarray}
Similarly,
\begin{equation}
\frac{\delta E[I_S(\gamma)\gamma g(\gamma)]}{\delta g(\gamma_0)} = \gamma_0p(\gamma_0),
\end{equation}
\begin{equation}
\frac{\delta E[I_S(\gamma)g(\gamma)]}{\delta g(\gamma_0)} =p(\gamma_0).
\end{equation}
Therefore,
\begin{equation}
\frac{\delta N[g(\gamma)]}{\delta g(\gamma_0)} = 2Q[g(\gamma)]\gamma_0p(\gamma_0),
\end{equation}
\begin{equation}
\frac{\delta D[g(\gamma)]}{\delta g(\gamma_0)} = 2g(\gamma_0)p(\gamma_0) - 2Q[g(\gamma)]\gamma_0p(\gamma_0) - 2Y[g(\gamma)]p(\gamma_0).
\end{equation}

Condition (\ref{eq:condition}) requires
\begin{equation}
\frac{\delta N[g(\gamma)]}{\delta g(\gamma_0)}D[g(\gamma)] = \frac{\delta D[g(\gamma)]}{\delta g(\gamma_0)}N[g(\gamma)].
\end{equation}
Substituting and simplifying, we obtain
\begin{equation}\label{ggamma0}
g(\gamma_0) = \frac{\gamma_0}{\eta} + \beta
\end{equation}
where
\begin{equation}\label{eta0L1}
 \eta = \frac{E[I_S(\gamma )\gamma g(\gamma)] + C_1^U}{E[I_S(\gamma )g^2(\gamma)] + C_0^U - \beta^2 + \sigma_v^2/A^2},
\end{equation}
\begin{equation}\label{beta0L1}
\beta = E [g(\gamma)] = E[I_S(\gamma )g(\gamma)] + C_0^U
\end{equation}
as the solution for (\ref{eq:condition}). Since (\ref{ggamma0}) holds $\forall \gamma_0 \in S$, we must have
\begin{equation}\label{ggamma}
g(\gamma) = \frac{\gamma}{\eta} + \beta, \quad \forall \gamma \in S.
\end{equation}

Substituting (\ref{ggamma}) into (\ref{eta0L1}) and (\ref{beta0L1}), we obtain
\begin{equation}\label{eta1}
\eta = \frac{C_1^U + C_2^S/\eta + C_1^S\beta}{C_0^U + C_2^S/\eta^2 + 2\beta C_1^S/\eta +\beta^2C_0^S - \beta^2 + \sigma_v^2/A^2},
\end{equation}
\begin{equation}\label{beta1}
\beta = C_0^U + C_1^S/\eta + \beta C_0^S
\end{equation}
where $C_0^S$, $C_1^S$ and $C_2^S$ are given by (\ref{Cdenotation}).

Solving for $\eta$ and $\beta$, we further simplify them to (\ref{etan}) and (\ref{beta}).

In summary, under the dynamic range constraint, the optimal $g(\cdot)$ that maximizes the SNDR is given by (\ref{ggamma}), where $\eta$ and $\beta$ are given by (\ref{etan}) and (\ref{beta}).

\section{Proof of \emph{Lemma 2}}\label{L2proof}

Comparing (\ref{Sdefinition}) with (\ref{ggamma}), we infer that $0 < \frac{\gamma}{\eta} + \beta < 1$ on $S$.
Therefore, the set $S$ must be a subset of $S^{\star}=(-\beta\eta,\eta-\beta\eta)$ if $\eta>0$ or $S^{\star}=(\eta-\beta\eta,-\beta\eta)$ if $\eta<0$.
The objective here is to determine the optimal $S$ such that the SNDR is maximized.

To further this objective, we rewrite SNDR as

\begin{equation}\label{SNDR1}
\mathrm{SNDR}^{-1} = \frac{E[g^2(\gamma)]-E^2[g(\gamma)]+\frac{\sigma_v^2}{A^2}}{E^2[\gamma g(\gamma)]}-1.
\end{equation}

Since $g(\gamma) = \frac{\gamma}{\eta} + \beta$ for $\gamma \in S$, we infer that
\begin{eqnarray}
E [g(\gamma)] &=& C_0^U + C_1^S/\eta + \beta C_0^S,\\
E [g^2(\gamma)] &=& C_0^U + C_2^S/\eta^2 + 2\beta C_1^S/\eta +\beta^2C_0^S,\\
E [\gamma g^2(\gamma)] &=& C_1^U + C_2^S/\eta + C_1^S\beta.
\end{eqnarray}

From (\ref{eta1}), we have

\begin{equation}
\begin{split}
\sigma_v^2/A^2 &= C_1^U/\eta + C_2^S/\eta^2 + C_1^S\beta/\eta + \beta^2\\
&- C_0^U - C_2^S/\eta^2 - 2\beta C_1^S/\eta -\beta^2C_0^S.
\end{split}
\end{equation}

Thus, (\ref{SNDR1}) can be further simplified to
\begin{eqnarray}
\mathrm{SNDR}^{-1} &=& \frac{C_1^U/\eta + C_2^S/\eta^2 + C_1^S\beta/\eta}{(C_1^U + C_2^S/\eta + C_1^S\beta)^2}-1\\
&=& (C_2^S + \eta C_1^U + \eta C_1^S\beta)^{-1}-1.
\end{eqnarray}

As a result, the original problem can be written as

\begin{eqnarray}
 &\smash{\displaystyle\max_{L, S, U}}& {\quad C_2^S + \eta C_1^U + \eta C_1^S\beta}\notag \\
 &s.t.& \quad L\cup S\cup U = R,  \\
 & & \quad S \subseteq (-\beta\eta, \eta-\beta\eta) \quad \mathrm{or} \quad (\eta-\beta\eta, -\beta\eta). \notag
\end{eqnarray}

Recall that $C_2^S$, $C_1^S$, $C_1^U$, $\eta$ and $\beta$ are all functions
of $L$, $S$ and $U$. Set
\begin{eqnarray}
& & R(L,S,U) \notag \\
 &=& C_2^S + \eta(L,S,U) \beta(L,S,U) C_1^S + \eta(L,S,U) C_1^U \notag\\
&=& \frac{N_0(L,S,U)}{D_0(L,S,U)}
\end{eqnarray}
where
\begin{equation}
\begin{split}
& \quad N_0(L,S,U)\\
& =  C_2^S C_0^UC_0^L +C_0^U(C_1^S)^2 + 2C_0^UC_1^UC_1^S + (C_1^U)^2 \\
& - C_0^S(C_1^U)^2 + C_2^S(1-C_0^S)\sigma_v^2/A^2 + (C_1^S)^2\sigma_v^2/A^2
\end{split}
\end{equation}
and
\begin{equation}
D_0(L,S,U)=C_0^UC_0^L + (1-C_0^S)\sigma_v^2/A^2.
\end{equation}

Differing from the traditional optimization problem, the variables here are sets.
Let us consider two cases.

\begin{case}
Suppose that ($L$, $S$, $U$) is a feasible solution. Let us consider a set $S_1 \subset S$ and
\begin{eqnarray}
S_1 &=& S - \Delta_1,\\
L_1 &=& L + \Delta_1,\\
U_1 &=& U
\end{eqnarray}
which means a subset of $S$ is partitioned into $L$.
\end{case}

\newcounter{mytempeqncnt4}
\begin{figure*}[t]
\normalsize
\setcounter{mytempeqncnt4}{\value{equation}}
\setcounter{equation}{108}
\begin{equation}\label{com1}
\begin{split}
& \quad \hat{N}_1(L_1,S_1,U_1)D_0(L,S,U)-N_0(L,S,U)D_1(L_1,S_1,U_1)\\
&= ((C_2^SC_0^U+(C_1^U)^2)D_0(L,S,U)-C_0^UN_0(L,S,U))C_0^{\Delta_1}+ (C_2^SD_0(L,S,U)-N_0(L,S,U))C_0^{\Delta_1}\sigma_v^2/A^2\\
&- C_0^UC_0^LD_0(L,S,U)C_2^{\Delta_1}-(1-C_0^S)D_0(L,S,U)C_2^{\Delta_1}\sigma_v^2/A^2+ 2C_0^UC_1^LD_0(L,S,U)C_1^{\Delta_1}-2C_1^SD_0(L,S,U)C_1^{\Delta_1}\sigma_v^2/A^2\\
&= 2(C_0^UC_1^L-C_1^S\sigma_v^2/A^2)(C_0^UC_0^L + (1-C_0^S)\sigma_v^2/A^2)C_1^{\Delta_1}-(C_1^S\sigma_v^2/A^2-C_0^UC_1^L)^2C_0^{\Delta_1}- (C_0^UC_0^L+(1-C_0^S)\sigma_v^2/A^2)^2C_2^{\Delta_1}\\
&\leq 2|C_0^UC_1^L-C_1^S\sigma_v^2/A^2|(C_0^UC_0^L + (1-C_0^S)\sigma_v^2/A^2)|C_1^{\Delta_1}|-(C_1^S\sigma_v^2/A^2-C_0^UC_1^L)^2C_0^{\Delta_1}- (C_0^UC_0^L+(1-C_0^S)\sigma_v^2/A^2)^2C_2^{\Delta_1}\\
&= 2\underbrace{|C_0^UC_1^L-C_1^S\sigma_v^2/A^2|(C_0^UC_0^L + (1-C_0^S)\sigma_v^2/A^2)}_{\geq 0}\underbrace{(|C_1^{\Delta_1}|-\sqrt{C_0^{\Delta_1}}\sqrt{C_2^{\Delta_1}})}_{\leq 0}\\
&- \underbrace{(|C_1^S\sigma_v^2/A^2-C_0^UC_1^L|\sqrt{C_0^{\Delta_1}}-(C_0^UC_0^L+(1-C_0^S)\sigma_v^2/A^2)\sqrt{C_2^{\Delta_1}})^2}_{\geq 0}\\
&\leq 0
\end{split}
\end{equation}

\setcounter{equation}{101}
\hrulefill
\vspace*{4pt}
\end{figure*}

\begin{figure}[htbp]
\centering
\subfigure[$L$, $S$, $U$]{
\label{case1:subfig:a} 
\includegraphics[width=3in]{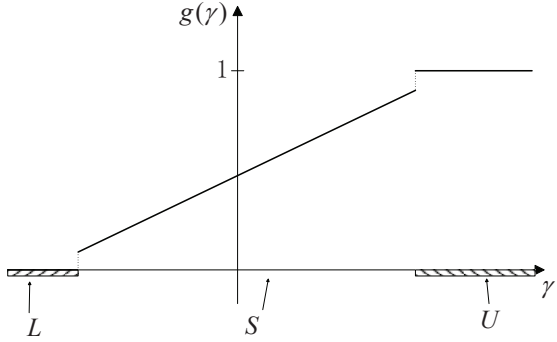}}
\hspace{0.2in}
\subfigure[$L_1=L+\Delta_1$, $S_1=S-\Delta_1$, $U_1=U$]{
\label{case1:subfig:b} 
\includegraphics[width=3in]{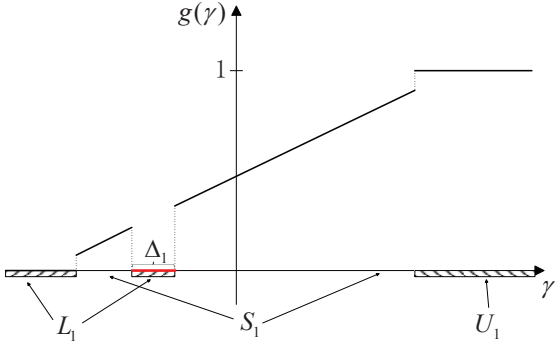}}
\caption{Example of \emph{Case 1}}
\label{case1:subfig} 
\end{figure}

Fig.~\ref{case1:subfig} demonstrates an example of \emph{Case 1}.
Then we have
\begin{equation}
R(L_1,S_1,U_1) = \frac{N_1(L_1,S_1,U_1)}{D_1(L_1,S_1,U_1)}
\end{equation}
where

\begin{equation}
\begin{split}
& \quad N_1(L_1,S_1,U_1)\\
&=  (C_2^S-C_2^{\Delta_1}) C_0^U(C_0^L+C_0^{\Delta_1}) +C_0^U(C_1^S-C_1^{\Delta_1})^2 \\
&+ 2C_0^U(C_1^S-C_1^{\Delta_1})C_1^U  - (C_0^S-C_0^{\Delta_1})(C_1^U)^2\\
&+ (C_2^S-C_2^{\Delta_1})(1-C_0^S+C_0^{\Delta_1})\sigma_v^2/A^2 \\
&+ (C_1^S-C_1^{\Delta_1})^2\sigma_v^2/A^2 + (C_1^U)^2\\
&= N_0(L,S,U) + C_2^SC_0^UC_0^{\Delta_1} - C_0^UC_0^LC_2^{\Delta_1}  \\
&-C_0^UC_2^{\Delta_1}C_0^{\Delta_1} - 2C_0^UC_1^SC_1^{\Delta_1}+C_0^U(C_1^{\Delta_1})^2 \\
&- 2C_0^UC_1^UC_1^{\Delta_1}+(C_1^U)^2C_0^{\Delta_1}-C_2^{\Delta_1}C_0^{\Delta_1}+(C_1^{\Delta_1})^2\\
&+ (C_2^SC_0^{\Delta_1}-(1-C_0^S)C_2^{\Delta_1}-2C_1^SC_1^{\Delta_1})\sigma_v^2/A^2
\end{split}
\end{equation}
and
\begin{equation}
\begin{split}
& \quad D_1(L_1,S_1,U_1) \\
&= C_0^U(C_0^L+C_0^{\Delta_1}) + (1-C_0^S+C_0^{\Delta_1})\sigma_v^2/A^2\\
&= D_0(L,S,U)+C_0^UC_0^{\Delta_1}+C_0^{\Delta_1}\sigma_v^2/A^2.
\end{split}
\end{equation}

Next, we would like to compare $R(L_1,S_1,U_1)$ and $R(L,S,U)$ to help us establish the optimal $S$ maximizing the SNDR.
However, it is a challenge to make the comparison directly since there are too many terms in the objective expression.
Here, we utilize a two-step comparison.

First, rewrite
\begin{equation}
\begin{split}
& \quad N_1(L_1,S_1,U_1)\\
&= N_0(L,S,U) + (C_2^SC_0^U+(C_1^U)^2)C_0^{\Delta_1} - C_0^UC_0^LC_2^{\Delta_1} \\
&- 2C_0^U(C_1^S+C_1^U)C_1^{\Delta_1}+C_0^U((C_1^{\Delta_1})^2 - C_2^{\Delta_1}C_0^{\Delta_1})\\
&+ (C_2^SC_0^{\Delta_1}-(1-C_0^S)C_2^{\Delta_1}-2C_1^SC_1^{\Delta_1})\sigma_v^2/A^2\\
&+((C_1^{\Delta_1})^2-C_2^{\Delta_1}C_0^{\Delta_1})\sigma_v^2/A^2\\
&\leq N_0(L,S,U) + (C_2^SC_0^U+(C_1^U)^2)C_0^{\Delta_1} \\
&- C_0^UC_0^LC_2^{\Delta_1} +2C_0^UC_1^LC_1^{\Delta_1}\\
&+ (C_2^SC_0^{\Delta_1}-(1-C_0^S)C_2^{\Delta_1}-2C_1^SC_1^{\Delta_1})\sigma_v^2/A^2\\
&= \hat{N}_1(L_1,S_1,U_1)
\end{split}
\end{equation}
where
\begin{equation}
(C_1^{\Delta_1})^2 \leq C_2^{\Delta_1}C_0^{\Delta_1}
\end{equation}
by the Cauchy-Schwartz inequality $(E[\theta\phi]) \leq E[\theta^2]E[\phi^2]$ with $\theta=\gamma I_{\Delta_1}(\gamma)$ and $\phi= I_{\Delta_1}(\gamma )$.

Next, we use $\hat{N}_1(L_1,S_1,U_1)$ instead of $N_1(L_1,S_1,U_1)$ to make the comparison.
Consider

\begin{equation}
\begin{split}
& \quad \frac{\hat{N}_1(L_1,S_1,U_1)}{D_1(L_1,S_1,U_1)}-\frac{N_0(L,S,U)}{D_0(L,S,U)}  \\
&=\frac{\hat{N}_1(L_1,S_1,U_1)D_0(L,S,U)-N_0(L,S,U)D_1(L_1,S_1,U_1)}{D_1(L_1,S_1,U_1)D_0(L,S,U)}
\end{split}
\end{equation}
where
\begin{equation}
\hat{N}_1(L_1,S_1,U_1)D_0(L,S,U)-N_0(L,S,U)D_1(L_1,S_1,U_1 \leq 0
\end{equation}
which is given by (\ref{com1}).

\setcounter{equation}{109}

Since both $D_1(L_1,S_1,U_1)$ and $D_0(L,S,U)$ are greater than zero, it can be concluded

\begin{equation}
R(L_1,S_1,U_1) \leq \frac{\hat{N}_1(L_1,S_1,U_1)}{D_1(L_1,S_1,U_1)} \leq R(L,S,U).
\end{equation}

\emph{Case 1} demonstrates that the SNDR will be decreased if any subset of $S$ is occupied by $L$. Let us consider another case.

\begin{case}
$S_2 \subset S$ and
\begin{eqnarray}
S_2 &=& S - \Delta_2,\\
L_2 &=& L,\\
U_2 &=& U  + \Delta_2
\end{eqnarray}
which means a subset of $S$ is partitioned into $U$.
\end{case}

\newcounter{mytempeqncnt3}
\begin{figure*}[t]
\normalsize
\setcounter{mytempeqncnt4}{\value{equation}}
\setcounter{equation}{120}
\begin{equation}\label{com2}
\begin{split}
& \quad \hat{N}_2(L_2,S_2,U_2)D_0(L,S,U)-N_0(L,S,U)D_2(L_2,S_2,U_2)\\
&= ((C_2^SC_0^L+(C_1^L)^2)D_0(L,S,U)-C_0^LN_0(L,S,U))C_0^{\Delta_2}+ (C_2^SD_0(L,S,U)-N_0(L,S,U))C_0^{\Delta_2}\sigma_v^2/A^2\\
&- C_0^UC_0^LD_0(L,S,U)C_2^{\Delta_2}-(1-C_0^S)D_0(L,S,U)C_2^{\Delta_2}\sigma_v^2/A^2+ 2C_0^LC_1^UD_0(L,S,U)C_1^{\Delta_2}-2C_1^SD_0(L,S,U)C_1^{\Delta_2}\sigma_v^2/A^2\\
&= 2(C_0^LC_1^U-C_1^S\sigma_v^2/A^2)(C_0^UC_0^L + (1-C_0^S)\sigma_v^2/A^2)C_1^{\Delta_1}- (C_1^S\sigma_v^2/A^2-C_0^LC_1^U)^2C_0^{\Delta_2}-(C_0^UC_0^L+(1-C_0^S)\sigma_v^2/A^2)^2C_2^{\Delta_2}\\
&\leq 2|C_0^LC_1^U-C_1^S\sigma_v^2/A^2|(C_0^UC_0^L + (1-C_0^S)\sigma_v^2/A^2)|C_1^{\Delta_2}|- (C_1^S\sigma_v^2/A^2-C_0^LC_1^U)^2C_0^{\Delta_2}-(C_0^UC_0^L+(1-C_0^S)\sigma_v^2/A^2)^2C_2^{\Delta_2}\\
&= 2\underbrace{|C_0^LC_1^U-C_1^S\sigma_v^2/A^2|(C_0^UC_0^L + (1-C_0^S)\sigma_v^2/A^2)}_{\geq 0}\underbrace{(|C_1^{\Delta_2}|-\sqrt{C_0^{\Delta_2}}\sqrt{C_2^{\Delta_2}})}_{\leq 0}\\
&- \underbrace{(|C_1^S\sigma_v^2/A^2-C_0^LC_1^U|\sqrt{C_0^{\Delta_2}}-(C_0^UC_0^L+(1-C_0^S)\sigma_v^2/A^2)\sqrt{C_2^{\Delta_2}})^2}_{\geq 0}\\
&\leq 0
\end{split}
\end{equation}

\setcounter{equation}{113}
\hrulefill
\vspace*{4pt}
\end{figure*}

\begin{figure}[htbp]
\centering
\subfigure[$L$, $S$, $U$]{
\label{case2:subfig:a} 
\includegraphics[width=3in]{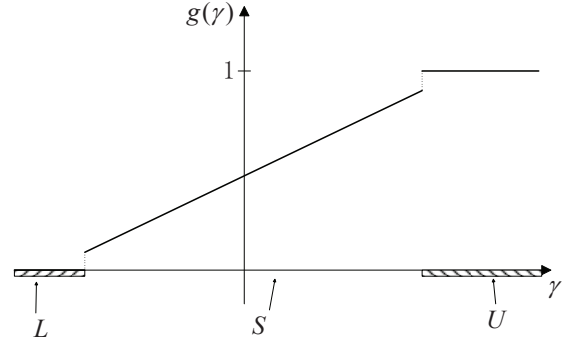}}
\hspace{0.2in}
\subfigure[$L_2=L$, $S_2=S-\Delta_2$, $U_2=U+\Delta_2$]{
\label{case2:subfig:b} 
\includegraphics[width=3in]{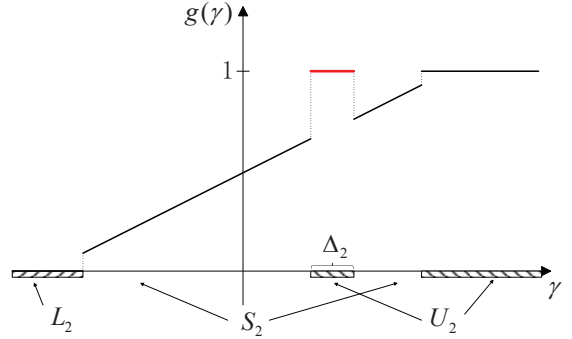}}
\caption{Example of \emph{Case 2}}
\label{case2:subfig} 
\end{figure}

Fig.~\ref{case2:subfig} demonstrates an example of \emph{Case 2}.
Then we have
\begin{equation}
R(L_2,S_2,U_2) = \frac{N_2(L_2,S_2,U_2)}{D_2(L_2,S_2,U_2)}
\end{equation}
where

\begin{equation}
\begin{split}
& \quad N_2(L_2,S_2,U_2)\\
&=  (C_2^S-C_2^{\Delta_2}) (C_0^U+C_0^{\Delta_2})C_0^L +(C_0^U+C_0^{\Delta_2})(C_1^S-C_1^{\Delta_2})^2\\
&+ 2(C_0^U+C_0^{\Delta_2})(C_1^U+C_1^{\Delta_2})(C_1^S-C_1^{\Delta_2})\\
&+ (C_1^U+C_1^{\Delta_2})^2 - (C_0^S-C_0^{\Delta_2})(C_1^U+C_1^{\Delta_2})^2\\
&+ (C_2^S-C_2^{\Delta_2})(1-C_0^S+C_0^{\Delta_2})\sigma_v^2/A^2 + (C_1^S-C_1^{\Delta_2})^2\sigma_v^2/A^2\\
&= N_0(L,S,U) + C_2^SC_0^LC_0^{\Delta_2} - C_0^UC_0^LC_2^{\Delta_2} -C_0^LC_2^{\Delta_2}C_0^{\Delta_2} \\
&- 2C_0^UC_1^SC_1^{\Delta_2}+C_0^U(C_1^{\Delta_2})^2+(C_1^S)^2C_0^{\Delta_2}-2C_1^SC_0^{\Delta_2}C_1^{\Delta_2}\\
&+ C_0^{\Delta_2}(C_1^{\Delta_2})^2+2C_0^UC_1^SC_1^{\Delta_2}+2C_1^UC_1^SC_0^{\Delta_2}\\
&+ 2C_1^SC_0^{\Delta_2}C_1^{\Delta_2}- 2C_0^UC_1^UC_1^{\Delta_2}-2C_0^U(C_1^{\Delta_2})^2\\
&- 2C_1^UC_0^{\Delta_2}C_1^{\Delta_2} - 2C_0^{\Delta_2}(C_1^{\Delta_2})^2+ 2C_1^UC_1^{\Delta_2}+(C_1^{\Delta_2})^2\\
&- 2C_0^SC_1^UC_1^{\Delta_2}-C_0^S(C_1^{\Delta_2})^2 + (C_1^U)^2C_0^{\Delta_2}+2C_1^UC_0^{\Delta_2}C_1^{\Delta_2}\\
&+ C_0^{\Delta_2}(C_1^{\Delta_2})^2+ (C_2^SC_0^{\Delta_2}-(1-C_0^S)C_2^{\Delta_2}-C_2^{\Delta_2}C_0^{\Delta_2}\\
&+ (C_1^{\Delta_2})^2v - 2C_1^SC_1^{\Delta_2})\sigma_v^2/A^2
\end{split}
\end{equation}
and
\begin{equation}
\begin{split}
&\quad D_2(L_2,S_2,U_2) \\
&= (C_0^U+C_0^{\Delta_2})C_0^L + (1-C_0^S+C_0^{\Delta_2})\sigma_v^2/A^2\\
&= D_0(L,S,U)+C_0^LC_0^{\Delta_2}+C_0^{\Delta_2}\sigma_v^2/A^2.
\end{split}
\end{equation}

In \emph{Case 2}, we also try to determine the difference between $R(L_2,S_2,U_2)$ and $R(L,S,U)$ by utilizing the two-step comparison.

First, rewrite
\begin{equation}
\begin{split}
&N_2(L_2,S_2,U_2)\\
&= N_0(L,S,U) + (C_2^SC_0^L+(C_1^U+C_1^S)^2)C_0^{\Delta_2} - C_0^UC_0^LC_2^{\Delta_2} \\
&+ 2C_0^LC_1^UC_1^{\Delta_2}+C_0^L((C_1^{\Delta_2})^2 - C_2^{\Delta_1}C_0^{\Delta_1})\\
&+ (C_2^SC_0^{\Delta_2}-(1-C_0^S)C_2^{\Delta_2}-2C_1^SC_1^{\Delta_2})\sigma_v^2/A^2\\
&+((C_1^{\Delta_2})^2-C_2^{\Delta_2}C_0^{\Delta_2})\sigma_v^2/A^2\\
&\leq N_0(L,S,U) + (C_2^SC_0^L+(C_1^L)^2)C_0^{\Delta_2} \\
&- C_0^UC_0^LC_2^{\Delta_2}+2C_0^LC_1^UC_1^{\Delta_2} \\
&+ (C_2^SC_0^{\Delta_2}-(1-C_0^S)C_2^{\Delta_2}-2C_1^SC_1^{\Delta_2})\sigma_v^2/A^2\\
&= \hat{N}_2(L_2,S_2,U_2)
\end{split}
\end{equation}
where
\begin{equation}
(C_1^{\Delta_2})^2 \leq C_2^{\Delta_2}C_0^{\Delta_2}
\end{equation}
by the Cauchy-Schwartz inequality.

Second, consider

\begin{equation}
\begin{split}
& \quad \frac{\hat{N}_2(L_2,S_2,U_2)}{D_2(L_2,S_2,U_2)}-\frac{N_0(L,S,U)}{D_0(L,S,U)}\\
&=\frac{\hat{N}_2(L_2,S_2,U_2)D_0(L,S,U)-N_0(L,S,U)D_2(L_2,S_2,U_2)}{D_2(L_2,S_2,U_2)D_0(L,S,U)}
\end{split}
\end{equation}
where
\begin{equation}
\hat{N}_2(L_2,S_2,U_2)D_0(L,S,U)-N_0(L,S,U)D_2(L_2,S_2,U_2)\leq 0
\end{equation}
which is given by (\ref{com2}).

\setcounter{equation}{121}

Since both $D_2(L_2,S_2,U_2)$ and $D_0(L,S,U)$ are greater than zero, it can be concluded

\begin{equation}
R(L_2,S_2,U_2) \leq \frac{\hat{N}_2(L_2,S_2,U_2)}{D_2(L_2,S_2,U_2)} \leq R(L,S,U).
\end{equation}

\emph{Case 2} demonstrates that the SNDR will be also decreased if any subset of $S$ is occupied by $U$.

Additionally, \emph{Case 1} and \emph{Case 2} also imply that the SNDR can be increased if $S$ can be enlarged by occupying the subsets of $L$ and $U$.
Thus, \emph{Lemma 2} holds and the optimal $S$ is implied to be $S^{\star}=(-\beta\eta,\eta-\beta\eta)$ if $\eta>0$ or $S^{\star} = (\eta-\beta\eta,-\beta\eta)$ if $\eta<0$.



\ifCLASSOPTIONcaptionsoff
  \newpage
\fi

\end{document}